\documentclass[11pt]{article}
\usepackage[T1]{fontenc}
\usepackage[utf8]{inputenc}
\usepackage[english]{babel}
\usepackage{lmodern}
\usepackage{microtype}
\usepackage{amsmath,amssymb,mathtools}
\usepackage{booktabs,array,tabularx}
\usepackage{graphicx}
\usepackage{float}
\usepackage{geometry}
\usepackage{enumitem}
\usepackage{needspace}
\setlist{leftmargin=*}
\usepackage{tikz}
\usetikzlibrary{positioning,fit,calc,arrows.meta,decorations.pathreplacing}
\usepackage{caption}
\usepackage{titlesec}
\titleformat{\section}{\bfseries\normalsize}{\thesection}{0.5em}{}
\titleformat{\subsection}{\bfseries\small}{\thesubsection}{0.5em}{}
\titlespacing*{\section}{0pt}{9pt}{4pt}
\titlespacing*{\subsection}{0pt}{7pt}{3pt}
\usepackage[hidelinks]{hyperref}
\hypersetup{
  pdftitle={The Interconnectedness Coefficient: A Semi-Local Graph-Theoretic Measure for Connector Vertices between Cohesive Network Regions},
  pdfauthor={Thomas Wiebringhaus},
  pdfsubject={Interconnectedness Coefficient (IC), semi-local graph-theoretic node measure for connector and bridging vertices between cohesive network regions, computation from triangle counts and common neighbors},
  pdfkeywords={Interconnectedness Coefficient, IC, semi-local centrality, local centrality, graph-theoretic node measure, connector vertex, bridging node, bridge node detection, structural bottleneck, cohesive network regions, local clustering coefficient, triangle-based graph measure, network centrality, vertex ranking, community-independent centrality, partition-free node measure, reference implementation, protein-protein interaction network, interface protein, scaffold protein, adaptor protein, adapter protein}
}
\usepackage{url}
\newcommand{\IC}{\mathrm{IC}}

\newtheorem{proposition}{Proposition}

\title{The Interconnectedness Coefficient:\\
A Semi-Local Graph-Theoretic Measure for Connector Vertices between Cohesive Network Regions}
\author{Thomas Wiebringhaus\\
\small ifes Institute for Empirical Research \& Statistics\\
\small FOM University of Applied Sciences, M\"unster, Germany\\
\small \href{mailto:thomas.wiebringhaus@fom.de}{thomas.wiebringhaus@fom.de}}
\date{12 September 2026}

\begin{document}
\maketitle

\begin{abstract}
The Interconnectedness Coefficient (IC) is a bounded semi-local graph-theoretic node measure designed to identify connector vertices between cohesive network regions. Such connector vertices, also referred to as bridging nodes, may mediate between locally cohesive regions even when they are neither hubs nor themselves highly clustered. The IC preferentially assigns high values to weakly clustered focal vertices whose adjacent vertices remain strongly clustered after exclusion of the focal connection. Candidate vertices are required to have degree at least two. The construction is partition-free, uses information within radius two, and requires no predefined community or module partition. The range and extremal properties of the score are derived analytically. Exact graph families isolate its maximal response to fully cohesive branches, its controlled response to a single cohesion defect, its invariance under a cohesion-free hub extension, and a sharp fragmentation threshold. A separate application to a Human Interactome Map reveals pronounced degree-dependent stabilization of IC values near the network's mean clustering level. This behavior follows directly from the multiplicative definition. If focal clustering tends to zero while mean leave-one-out cohesion in the neighborhood stabilizes, the IC converges to that neighborhood-cohesion level. Among the highly ranked IC vertices are proteins with established interface, scaffold, and adaptor roles in molecular complexes. The IC is therefore positioned as a semi-local connector measure for cohesive network regions.
\end{abstract}


\section{Motivation and Contribution}
The Interconnectedness Coefficient (IC) serves to identify connector vertices between cohesive network regions from a bounded radius-two neighborhood. It therefore belongs to the broader domain of vertex ranking and centrality measures while targeting a specific structural role. Centrality does not refer to a single structural property. Different node scores encode shortest-path mediation, diffusion potential, local cohesion, $k$-core position, or access to otherwise separated regions \cite{Freeman1978,Lu2016}. The guiding question of the IC is:

\begin{quote}\itshape
Which vertices connect cohesive local regions while remaining only weakly clustered themselves?
\end{quote}

The target object is a connector vertex that links locally cohesive regions without requiring communities or modules to be specified in advance. For example, in a protein--protein interaction network (PPI), proteins are represented as vertices and reported interactions as edges. Functional complexes or modules may appear as cohesive regions. Proteins or small complexes located between such regions then constitute candidate mediators. The graph-theoretic target is therefore not an abstract notion of ``importance'' but a specific topology: a sparsely interconnected focal neighborhood with cohesive structure immediately beyond the neighbors of the focal vertex. Figure~\ref{fig:calibration} illustrates this semi-local signature.

\begin{figure}[H]
\centering
\begin{tikzpicture}[x=1cm,y=1cm,every node/.style={font=\scriptsize}]
\path[use as bounding box] (-3.4,-1.35) rectangle (3.4,1.35);
\node[circle,draw,inner sep=1.3pt] (v) at (0,0) {$v$};
\node at (0,.62) {\Large ?};
\node[circle,draw,inner sep=1.1pt] (a) at (-1.25,0) {};
\node[circle,draw,inner sep=1.1pt] (b) at (1.25,0) {};
\draw (a)--(v)--(b);
\foreach \x/\y in {-2.6/.48,-2.6/-.48,-1.9/.76,-1.9/-.76} \node[circle,draw,inner sep=1.0pt] at (\x,\y) {};
\draw (-2.6,.48)--(-2.6,-.48)--(-1.9,-.76)--(-1.9,.76)--cycle;
\draw (-2.6,.48)--(-1.9,-.76) (-2.6,-.48)--(-1.9,.76);
\draw (a)--(-2.6,.48) (a)--(-2.6,-.48) (a)--(-1.9,.76) (a)--(-1.9,-.76);
\foreach \x/\y in {2.6/.48,2.6/-.48,1.9/.76,1.9/-.76} \node[circle,draw,inner sep=1.0pt] at (\x,\y) {};
\draw (2.6,.48)--(2.6,-.48)--(1.9,-.76)--(1.9,.76)--cycle;
\draw (2.6,.48)--(1.9,-.76) (2.6,-.48)--(1.9,.76);
\draw (b)--(2.6,.48) (b)--(2.6,-.48) (b)--(1.9,.76) (b)--(1.9,-.76);
\node at (0,1.10) {Graph-theoretic calibration};
\node at (0,-1.10) {Connector candidate};
\end{tikzpicture}
\caption{Illustration of the structural role targeted by the IC. The focal vertex $v$ has a weakly interconnected immediate neighborhood, while cohesive structures lie beyond its adjacent gateway vertices. The IC tests this pattern without requiring a predefined community partition.}
\label{fig:calibration}
\end{figure}
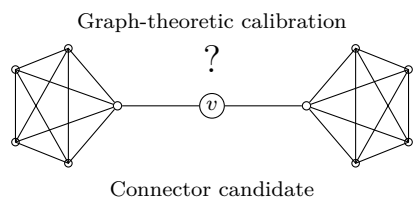

The IC is semi-local because its information radius is two. It requires neither a module partition nor the prior condition that a candidate be an articulation vertex. A high IC value therefore does not imply articulation or guarantee that deleting the vertex disconnects the graph. Longer alternative paths may reconnect regions that appear separated within the radius-two neighborhood.

This bounded information requirement also distinguishes the IC computationally from global shortest-path measures such as betweenness centrality. Exact betweenness requires information about shortest paths across the graph, whereas the IC can be computed from neighborhoods, triangle counts, and common-neighbor intersections. The IC is not intended as an approximation to betweenness. It targets a different structural signature using semi-local information, which is attractive when large graphs make global path-based calculations expensive \cite{Freeman1978,Lu2016}.

The contribution has four components. First, a self-contained definition makes both the candidate set and the leave-one-out clustering quantity behind each adjacent branch explicit. Second, the range of the IC and its equality conditions are derived. Third, exact graph families isolate the response of the score to branch multiplicity, a single cohesion defect, a cohesion-free hub extension, and controlled fragmentation. Fourth, an application to a protein--protein interaction network provides an empirical layer with biologically interpretable findings and degree-dependent stabilization of the IC. The IC construction was initially developed in a bioinformatic network analysis \cite{Wiebringhaus2009} and subsequently presented as a distinct approach to intermodular connections \cite{WiebringhausBrinck2014}. The present work formalizes it as an independent graph-theoretic measure, calibrates its structural behavior analytically, and renders it algorithmically operational.

\section{Related Work and Research Questions}
Centrality does not refer to a single structural property. Freeman's classical clarification distinguishes several meanings of central position. Graph-theoretic and flow-based treatments likewise emphasize that a node score must be interpreted in relation to the structural role or process that it represents \cite{Freeman1978,Borgatti2005,BorgattiEverett2006}. Reviews span local, semi-local, and global objectives for vertex ranking \cite{Lu2016}. A recent encyclopedia catalogues more than 400 distinct node metrics, which further illustrates this diversity \cite{Shvydun2025}.

One major line of work identifies connectors relative to an explicit community structure. Functional cartography distinguishes within-module connectivity from participation across modules \cite{Guimera2005}. Fuzzy bridgeness extends this idea to overlapping community memberships \cite{Nepusz2008}. The Gateway Coefficient emphasizes rare or critical inter-module connections \cite{VargasWahl2014}, while Modular Centrality separates local and global influence after a community structure has been identified \cite{Ghalmane2019}. These approaches are natural when a partition or membership structure is part of the problem. The IC instead asks whether cohesive regions can be detected around a candidate directly from semi-local topology.

A related line of work targets bridging positions without relying exclusively on a community partition. Bridging Centrality combines a local bridging component with global betweenness information \cite{Hwang2008}. Localized Bridging Centrality replaces the global component with egocentric betweenness centrality and thereby shifts the calculation toward local information \cite{NandaKotz2008}. Bridging Node Centrality combines normalized route-betweenness with a distance-based bridgeness component \cite{Liu2019}. Comparative studies show that such bridging and participation-oriented measures can produce rankings that differ substantially from classical centralities \cite{Rajeh2021}.

Clustering information also appears in several local and semi-local ranking methods. Local Centrality with a Coefficient modifies a local centrality through a clustering-based factor \cite{Zhao2017}. Normalized Local Centrality incorporates local structural attributes together with clustering information from neighboring vertices \cite{Zhao2018}. Degree and Neighborhood Information Centrality combines degree with the local clustering coefficients of immediate neighbors \cite{Zhao2024}. These constructions show that both low focal clustering and neighborhood clustering have become recurring ingredients in node ranking. The IC differs in the specific relationship it requires between them: low clustering at the focal vertex must coincide with high leave-one-out cohesion behind its adjacent vertices.

Recent work sharpens this distinction. Dong et al. derive Local Dispersion Centrality from degree and local clustering and then combine it with global betweenness in a hybrid method for identifying structurally critical vertices \cite{Dong2026}. Meghanathan combines the complement of local clustering with degree in a principal-component model motivated by node betweenness and core--periphery structure \cite{Meghanathan2026}. Both use sparse focal neighborhoods as evidence of structural relevance. Neither, however, requires cohesive structure to persist behind each neighboring gateway. Wang et al. address critical and bridging components through hierarchical community information from local to global scales \cite{Wang2025}. Their construction depends on an inferred community hierarchy, whereas the IC remains partition-free.

Particularly close to the IC in structural objective and information radius is Meghanathan's partition-free Neighborhood-based Bridge Node Centrality (NBNC) tuple \cite{Meghanathan2021}. NBNC uses two-hop neighborhood information and characterizes a bridge vertex through the number of connected components in its neighborhood graph, the algebraic connectivity ratio of that graph, and vertex degree. The IC asks a different semi-local question. A sparsely interconnected focal neighborhood is not sufficient. Cohesive structure must also persist behind individual neighbors after the focal connection has been excluded. The center of a star therefore receives exactly zero under the IC even though its neighborhood graph decomposes into singleton components. The defining IC signature is the combination of low focal clustering with high branch-side leave-one-out cohesion.

Berahmand, Bouyer, and Samadi provide another close methodological reference point \cite{Berahmand2018}. Their semi-local score combines a negative contribution from focal clustering with degree and a positive contribution from clustering among second-order neighbors. The method is designed to identify influential spreaders and is evaluated with SI and SIR diffusion models. The overlap in constituent information is important for positioning the IC. The objectives are nevertheless different. Their validation target is spreading performance, while the IC targets mediation between cohesive regions and does not define success through diffusion.

The paper is organized around five research questions:
\begin{description}[leftmargin=3.8em,style=nextline]
\item[RQ1 --- Identification.] Under which semi-local conditions does the IC attain its maximum and select a connector between cohesive branches?
\item[RQ2 --- Structural robustness.] How does the IC respond to branch multiplicity, a local cohesion defect, and a cohesion-free hub extension, and how does this response differ from simple degree-weighted low-clustering scores?
\item[RQ3 --- Fragmentation threshold.] At which exact structural threshold does the IC cease to prefer an external connector because the nominal branch itself fragments into several locally complete subbranches?
\item[RQ4 --- Empirical stabilization.] Why does the application to the human protein-interaction network exhibit a narrowing degree--IC distribution at high degree, and how is its limiting level related to network clustering?
\item[RQ5 --- Biological application.] How are highly ranked IC vertices positioned with respect to cohesive molecular structures in a real protein--protein interaction network?
\end{description}

\section{Mathematical Definition}
Formally, the IC is a semi-local, partition-free node measure for the topology of a connector vertex between cohesive regions. Let $G=(V,E)$ be a finite simple undirected graph. For $v\in V$, let the \emph{neighborhood} of $v$ be
\[
N(v)=\{u\in V:\{u,v\}\in E\},
\]
and let $d(v)=|N(v)|$ denote its degree. For a vertex set $S\subseteq V$, let $e(S)$ denote the number of edges induced by $S$. Using the standard local clustering coefficient \cite{Watts1998}, define
\begin{equation}
C(v)=
\begin{cases}
\dfrac{2e(N(v))}{d(v)(d(v)-1)}, & d(v)\ge 2,\\[5pt]
0, & d(v)<2.
\end{cases}
\end{equation}
Equivalently, for $d(v)\ge2$,
\[
C(v)=\frac{e(N(v))}{\binom{d(v)}{2}},
\]
so the denominator is the number of possible undirected edges among the $d(v)$ neighbors of $v$. For vertices of degree zero or one, we set $C(v)=0$ because there is no pair of neighbors on which a nontrivial local clustering value can be based.

The candidate set is
\begin{equation}
V_{\ge2}=\{v\in V:d(v)\ge2\}.
\end{equation}
Vertices outside $V_{\ge2}$ are not ranked. This is the minimal structural restriction required for a connector to possess at least two incident directions.

For an ordered adjacent pair $(u,v)$, define the leave-one-out clustering of $u$ relative to $v$ by
\begin{equation}
C_{u\setminus v}=
\begin{cases}
\dfrac{2e(N(u)\setminus\{v\})}{(d(u)-1)(d(u)-2)}, & d(u)\ge3,\\[5pt]
0, & d(u)<3.
\end{cases}
\end{equation}
Equivalently, for $d(u)\ge3$,
\[
C_{u\setminus v}=\frac{e(N(u)\setminus\{v\})}{\binom{d(u)-1}{2}},
\]
where $\binom{d(u)-1}{2}$ is the number of possible edges among the $d(u)-1$ neighbors that remain after $v$ is excluded. Removing the focal vertex from this branch-side neighborhood prevents the connection $uv$ from artificially lowering the cohesion assigned to an otherwise cohesive branch. Thus, $C_{u\setminus v}$ measures how cohesive the local structure behind $u$ remains after $v$ has been excluded from $u$'s neighborhood.

For each $v\in V_{\ge2}$, define the mean branch-side leave-one-out cohesion
\begin{equation}
\overline C_{\mathrm{LOO}}(v)=\frac{1}{d(v)}\sum_{u\in N(v)}C_{u\setminus v}.
\label{eq:cloo}
\end{equation}
The Interconnectedness Coefficient is then
\begin{equation}
\boxed{\displaystyle
\IC(v)=\bigl(1-C(v)\bigr)\overline C_{\mathrm{LOO}}(v).}
\label{eq:ic_product}
\end{equation}
The definition is a product of two quantities. The focal factor $1-C(v)$ rewards low clustering among the immediate neighbors of $v$. The second factor $\overline C_{\mathrm{LOO}}(v)$ rewards high mean cohesion behind those neighbors. The use of a mean rather than a sum is deliberate. Adding further branches of equal cohesion does not, by itself, increase the scale of the score.

\begin{proposition}[Range and maximum]
For every $v\in V_{\ge2}$, $0\le \IC(v)\le1$. Moreover, $\IC(v)=1$ if and only if $C(v)=0$ and $C_{u\setminus v}=1$ for every $u\in N(v)$.
\end{proposition}
\textit{Proof.} Every clustering term lies in $[0,1]$. Hence both $1-C(v)$ and $\overline C_{\mathrm{LOO}}(v)$ lie in $[0,1]$, as does their product. A product of two numbers in this interval equals one if and only if both factors equal one. The mean equals one if and only if every summand equals one.

The equality case makes the intended semi-local geometry explicit. The focal neighbors are mutually nonadjacent, while the branch-side neighborhoods that remain after exclusion of $v$ are complete. Figure~\ref{fig:ic_maximum} shows the smallest symmetric example.

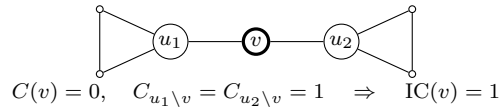
\begin{figure}[H]
\centering
\begin{tikzpicture}[scale=.78,every node/.style={font=\scriptsize}]
\node[circle,draw,very thick,inner sep=1.4pt] (v) at (0,0) {$v$};
\node[circle,draw,inner sep=1.1pt] (u1) at (-1.45,0) {$u_1$};
\node[circle,draw,inner sep=1.1pt] (u2) at (1.45,0) {$u_2$};
\node[circle,draw,inner sep=.9pt] (l1) at (-2.65,.55) {};
\node[circle,draw,inner sep=.9pt] (l2) at (-2.65,-.55) {};
\node[circle,draw,inner sep=.9pt] (r1) at (2.65,.55) {};
\node[circle,draw,inner sep=.9pt] (r2) at (2.65,-.55) {};
\draw (v)--(u1) (v)--(u2);
\draw (u1)--(l1)--(l2)--(u1);
\draw (u2)--(r1)--(r2)--(u2);
\node at (0,-.90) {$C(v)=0,\quad C_{u_1\setminus v}=C_{u_2\setminus v}=1\quad\Rightarrow\quad \IC(v)=1$};
\end{tikzpicture}
\caption{Exact local geometry for the maximum $\IC(v)=1$. The neighbors $u_1$ and $u_2$ are not adjacent to one another, but each opens into a complete branch-side neighborhood after the focal vertex $v$ is excluded.}
\label{fig:ic_maximum}
\end{figure}

\textbf{Corollary (triangle-free graphs).} If $G$ is triangle-free, then $\IC(v)=0$ for all $v\in V_{\ge2}$. In particular, the IC vanishes identically on trees and simple bipartite graphs.

\textit{Proof.} In a triangle-free graph, no neighborhood induces an edge. Hence $C(v)=0$ for every vertex and likewise $C_{u\setminus v}=0$ for every ordered adjacent pair $(u,v)$. The definition of the IC immediately yields $\IC(v)=0$.

A high IC value does not imply that $v$ is an articulation vertex. The score uses only radius-two information, so longer alternative paths may reconnect locally distinct regions after removal of $v$. Articulation is therefore a possible global consequence of the targeted topology, not a prerequisite of the IC definition.

\subsection{Computation, Complexity, and Reference Implementation}
The definition can be reduced directly to triangle counts and intersections of neighborhoods. Let
\[
T(u)=e(N(u))
\]
denote the number of triangles containing vertex $u$, and for an edge $\{u,v\}\in E$, let
\[
t(u,v)=|N(u)\cap N(v)|
\]
denote the number of common neighbors of $u$ and $v$. Removing $v$ from $N(u)$ removes exactly the $t(u,v)$ edges joining $v$ to common neighbors, so
\[
e(N(u)\setminus\{v\})=T(u)-t(u,v).
\]
Thus, for $d(u)\ge3$,
\begin{equation}
C_{u\setminus v}=\frac{2\bigl(T(u)-t(u,v)\bigr)}{(d(u)-1)(d(u)-2)}.
\label{eq:loo_triangle}
\end{equation}
This expression is algebraically identical to the leave-one-out definition in Eq.~(3). It can be evaluated from precomputed triangle and common-neighbor counts without explicitly constructing the reduced set $N(u)\setminus\{v\}$ for every ordered adjacent pair. Moreover,
\[
T(u)=\frac12\sum_{v\in N(u)}t(u,v),
\]
because every triangle containing $u$ contributes to exactly two terms in this sum, one through each of its two edges incident to $u$. A direct algorithm therefore consists of four steps: (1) determine neighborhood sets and degrees, (2) for each edge $\{u,v\}$ count the intersection $N(u)\cap N(v)$, (3) compute $T(u)$ and all leave-one-out terms, and (4) average these terms for each candidate and multiply by $1-C(v)$. If neighborhoods are stored as hash sets and the smaller neighborhood is probed against the larger, the expected dominant running time is
\begin{equation}
O\!\left(\sum_{\{u,v\}\in E}\min\{d(u),d(v)\}\right),
\end{equation}
because the smaller of the two neighborhood sets can be scanned while membership is tested in the larger set. Thus, an edge joining vertices of degrees 5 and 100 contributes expected intersection work of order 5 rather than 100. This dominant term is accompanied by additional linear work in $|V|+|E|$. The memory requirement remains $O(|V|+|E|)$.

\begin{center}
\fbox{\begin{minipage}{0.94\linewidth}
\small
\textbf{Pseudocode: Computation of the Interconnectedness Coefficient}
\begin{enumerate}[leftmargin=1.5em,itemsep=1pt,topsep=2pt]
\item Determine $N(v)$ and $d(v)$ for all $v\in V$.
\item For each edge $\{u,v\}\in E$, determine the number of common neighbors $t(u,v)=|N(u)\cap N(v)|$.
\item Compute $T(u)=\tfrac12\sum_{v\in N(u)}t(u,v)$ and derive all $C_{u\setminus v}$ according to Eq.~\eqref{eq:loo_triangle}.
\item For each $v\in V_{\ge2}$, average the terms $C_{u\setminus v}$ and set $\IC(v)=(1-C(v))\,d(v)^{-1}\sum_{u\in N(v)}C_{u\setminus v}$.
\end{enumerate}
\textbf{Output:} $\{(v,\IC(v)):v\in V_{\ge2}\}$.
\end{minipage}}
\end{center}

A minimal numerical example first illustrates how the two factors of the score interact with the averaging operation before the exact graph families vary the same structural motif systematically.

\textbf{Minimal numerical example.} If a focal vertex $v$ has degree three, $C(v)=0$, and branch-side leave-one-out values $1,1,0$, then
\[
\IC(v)=(1-0)\frac{1+1+0}{3}=\frac23.
\]
This exact local configuration is realized below by the hub-extended bridge family $H_{p,q,r}$. An executable reference implementation for simple undirected graphs based on NetworkX \cite{Hagberg2008} is included in the source package as \path{interconnectedness_coefficient.py}. It returns values exclusively for the candidate set $V_{\ge2}$.

\section{Exact Calibration Graphs}
The following exact graph families show how the IC responds to controlled structural changes. The sequence increases structural complexity step by step. A star first shows that degree alone does not create an IC signal. A two-sided bridge realizes the maximum. A multi-branch construction tests invariance to the number and size of complete branches. A single missing edge introduces a minimal cohesion defect. A hub extension tests whether a large cohesion-free hub can displace the intended connector. Finally, windmill branches fragment the nominally cohesive regions themselves.

\subsection{Star, Two-Sided Bridge, and Multiple Cohesive Branches}
The \emph{star graph} $S_{r+1}$ consists of a center $s$ and $r\ge2$ leaves, giving $r+1$ vertices in total. The restriction $r\ge2$ ensures that the center belongs to the candidate set. All local and leave-one-out clustering coefficients are zero. Hence
\begin{equation}
\IC(s)=0.
\end{equation}
The center has high degree relative to all other vertices in the star, but it does not connect cohesive regions. Degree-one leaves lie outside the candidate set. Figure~\ref{fig:star} shows the construction.

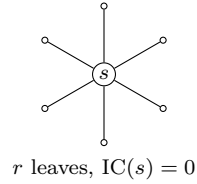
\begin{figure}[H]
\centering
\begin{tikzpicture}[scale=.86,every node/.style={font=\scriptsize}]
\node[circle,draw,inner sep=1.2pt] (s) at (0,0) {$s$};
\foreach \ang in {30,90,150,210,270,330} {
  \node[circle,draw,inner sep=.8pt] (n\ang) at (\ang:1.05) {};
  \draw (s)--(n\ang);
}
\node at (0,-1.45) {$r$ leaves, $\IC(s)=0$};
\end{tikzpicture}
\caption{Star graph $S_{r+1}$. The center $s$ can have arbitrarily many leaves, but no branch contains cohesive local structure. Consequently, the IC of the center is zero.}
\label{fig:star}
\end{figure}

The \emph{two-sided bridge graph} $B_{p,q}$ is constructed from vertex-disjoint cliques $K_p$ and $K_q$ with gateway vertices $a$ and $b$. A new vertex $x$ is connected to $a$ and $b$. For $p,q\ge3$, $C(x)=0$ and $C_{a\setminus x}=C_{b\setminus x}=1$. Therefore
\begin{equation}
\IC(x)=1.
\end{equation}
Figure~\ref{fig:bridge} shows the bridge graph and the two complete branch-side regions.

\begin{figure}[H]
\centering
\begin{tikzpicture}[x=1cm,y=1cm,every node/.style={font=\scriptsize}]
\node[circle,draw,inner sep=.7pt] (l1) at (-3.0,.58) {};
\node[circle,draw,inner sep=.7pt] (l2) at (-2.4,.58) {};
\node[circle,draw,inner sep=.7pt] (l3) at (-3.0,-.58) {};
\node[circle,draw,inner sep=.7pt] (l4) at (-2.4,-.58) {};
\node[circle,draw,inner sep=1.0pt] (a) at (-1.65,0) {$a$};
\draw (l1)--(l2)--(l3)--(l4)--cycle (l1)--(l3) (l1)--(l4) (l2)--(l4);
\draw (a)--(l1) (a)--(l2) (a)--(l3) (a)--(l4);
\node[draw,rounded corners,fit=(a)(l1)(l2)(l3)(l4),inner sep=3pt,label=above:$K_p$] {};
\node[circle,draw,very thick,inner sep=1.2pt] (x) at (0,0) {$x$};
\node[circle,draw,inner sep=1.0pt] (b) at (1.65,0) {$b$};
\node[circle,draw,inner sep=.7pt] (r1) at (2.4,.58) {};
\node[circle,draw,inner sep=.7pt] (r2) at (3.0,.58) {};
\node[circle,draw,inner sep=.7pt] (r3) at (2.4,-.58) {};
\node[circle,draw,inner sep=.7pt] (r4) at (3.0,-.58) {};
\draw (r1)--(r2)--(r3)--(r4)--cycle (r1)--(r3) (r1)--(r4) (r2)--(r4);
\draw (b)--(r1) (b)--(r2) (b)--(r3) (b)--(r4);
\node[draw,rounded corners,fit=(b)(r1)(r2)(r3)(r4),inner sep=3pt,label=above:$K_q$] {};
\draw (a)--(x)--(b);
\node at (0,-1.15) {$\IC(x)=1$};
\end{tikzpicture}
\caption{Two-sided bridge graph $B_{p,q}$. The connector $x$ is unclustered, while both gateway vertices open into complete branch-side neighborhoods after $x$ is excluded. This realizes the exact maximum $\IC(x)=1$.}
\label{fig:bridge}
\end{figure}
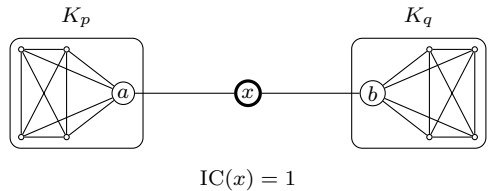

\paragraph{Why leave-one-out cohesion is required.}
The graph $B_{p,q}$ also makes the design choice in Eq.~(3) explicit. If ordinary clustering of a gateway were used instead, the focal connector $x$ would itself reduce the measured cohesion of an otherwise complete branch. For the left gateway $a$, its $p$ neighbors consist of $x$ and the $p-1$ other vertices of $K_p$. Since $x$ has no edges to those clique vertices,
\begin{equation}
C(a)=\frac{\binom{p-1}{2}}{\binom{p}{2}}=\frac{p-2}{p}<1,
\end{equation}
and analogously $C(b)=(q-2)/q$. A naive variant based on ordinary gateway clustering,
\[
\widetilde I(x)=(1-C(x))\frac{C(a)+C(b)}{2},
\]
would therefore assign
\begin{equation}
\widetilde I(x)=\frac12\left(\frac{p-2}{p}+\frac{q-2}{q}\right)<1,
\end{equation}
even though both branches are internally complete. Its value would also depend on the arbitrary branch orders $p$ and $q$. In contrast, excluding the focal connection gives $C_{a\setminus x}=C_{b\setminus x}=1$ and hence $\IC(x)=1$ independently of clique size. The leave-one-out operation prevents the connector edge itself from being interpreted as a cohesion defect behind the gateway.

The same result holds for an arbitrary number of complete branches. The \emph{multi-branch graph} $M_{m_1,\ldots,m_k}$ extends the two-sided bridge to $k$ cohesive branches. Formally, let $M_{m_1,\ldots,m_k}$ contain vertex-disjoint cliques $K_{m_1},\ldots,K_{m_k}$ with $m_j\ge3$, one gateway vertex $u_j$ in each clique, and a new vertex $v$ adjacent exactly to $u_1,\ldots,u_k$, where $k\ge2$.

\begin{proposition}[Calibration for two and multiple branches]
For the connector $v$ in $M_{m_1,\ldots,m_k}$,
\[
C(v)=0,\qquad C_{u_j\setminus v}=1\quad(1\le j\le k),\qquad \IC(v)=1.
\]
Thus, neither the number of branches nor clique size increases the score once all branches already satisfy the maximum-cohesion conditions.
\end{proposition}
\textit{Proof.} Gateway vertices from distinct cliques are mutually nonadjacent, so $C(v)=0$. After removing $v$ from $N(u_j)$, the remaining $m_j-1$ vertices induce $K_{m_j-1}$, yielding leave-one-out clustering equal to one. The mean of $k$ unit terms is one.

\paragraph{A single cohesion defect.}
The clique construction defines the exact maximum but does not imply that empirical branches must be complete. A minimal perturbation already shows a gradual response. Starting from $B_{p,q}$ with $p\ge4$, remove one edge between two non-gateway vertices of the left clique while leaving the right clique unchanged. The focal vertex still satisfies $C(x)=0$, whereas
\begin{equation}
C_{a\setminus x}=1-\frac{1}{\binom{p-1}{2}},
\qquad
C_{b\setminus x}=1.
\end{equation}
Consequently,
\begin{equation}
\IC(x)=1-\frac{1}{2\binom{p-1}{2}}.
\end{equation}
For $p=4$ this gives $\IC(x)=5/6\approx0.833$, while for $p=10$ it gives $\IC(x)=71/72\approx0.986$. Thus a single missing branch edge lowers the ideal score by a controlled amount rather than destroying the connector signal. The effect also decreases with branch size because the same missing edge represents a smaller fraction of the possible branch-side connections.

\subsection{Hub-Extended Bridge and Separation from Degree-Weighted Low Clustering}
Form the \emph{hub-extended bridge family} $H_{p,q,r}$ from $B_{p,q}$ by adding a vertex $s$ with $r$ leaf neighbors and the edge $sx$, where $p,q\ge3$ and $r\ge1$. The connector $x$ now has three incident directions. The two clique gateways remain maximally cohesive after exclusion of $x$, while the direction through the hub $s$ has leave-one-out clustering zero. Hence
\begin{equation}
\IC(x)=\frac{2}{3}.
\end{equation}
Figure~\ref{fig:hub_extension} shows the two cohesive branches together with the cohesion-free hub arm.

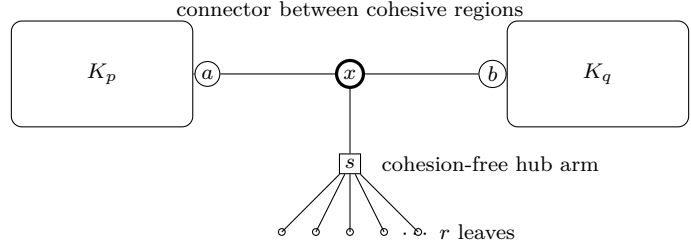
\begin{figure}[H]
\centering
\begin{tikzpicture}[scale=.82,every node/.style={font=\scriptsize}]
\node[draw,rounded corners,minimum width=2.4cm,minimum height=1.4cm] (L) at (-4,0) {$K_p$};
\node[circle,draw,fill=white,inner sep=1.3pt] (a) at (-2.3,0) {$a$};
\node[circle,draw,very thick,inner sep=1.5pt] (x) at (0,0) {$x$};
\node[circle,draw,fill=white,inner sep=1.3pt] (b) at (2.3,0) {$b$};
\node[draw,rounded corners,minimum width=2.4cm,minimum height=1.4cm] (R) at (4,0) {$K_q$};
\draw (L)--(a)--(x)--(b)--(R);
\node[draw,rectangle,inner sep=2pt] (s) at (0,-1.45) {$s$};
\draw (x)--(s);
\foreach \xx in {-1.1,-.55,0,.55,1.1} {\node[circle,draw,inner sep=.9pt] at (\xx,-2.55) {}; \draw (s)--(\xx,-2.55);}
\node at (1.75,-2.55) {$\cdots\ r$ leaves};
\node at (0,1.02) {connector between cohesive regions};
\node[right=4pt of s] {cohesion-free hub arm};
\end{tikzpicture}
\caption{Hub-extended bridge family $H_{p,q,r}$. The connector $x$ joins two complete branches and the hub $s$. Increasing $r$ raises the degree of $s$ without creating cohesive structure behind its leaf neighbors.}
\label{fig:hub_extension}
\end{figure}

For a gateway vertex attached to a clique of order $m$, the IC can be written directly. If $m\ge4$, its focal factor is $1-C=2/m$. Among its $m$ adjacent directions, $m-1$ lead into the complete clique remainder and contribute leave-one-out value one, while the direction toward $x$ contributes zero. Hence its mean leave-one-out cohesion is $(m-1)/m$, and
\[
\IC(\text{gateway in }K_m)=\frac{2(m-1)}{m^2}\qquad(m\ge4).
\]
For $m=3$, the leave-one-out neighborhood in the relevant boundary case has only one vertex and the defined value is zero.

\begin{proposition}[Hub-extension robustness]
For all integers $p,q\ge3$ and $r\ge1$, $x$ is the unique global maximizer of the IC over $V_{\ge2}$ in $H_{p,q,r}$.
\end{proposition}

Table~\ref{tab:hub_scores} lists the exact IC values of all vertex types needed for the comparison.

\begin{table}[H]
\centering
\caption{Exact IC scores for the vertex types in the hub-extended bridge family $H_{p,q,r}$.}
\label{tab:hub_scores}
\begin{tabular}{lc}
\toprule
Vertex type & IC \\
\midrule
Connector $x$ & $2/3$ \\
Hub $s$ & $0$ \\
Gateway $a$ in $K_p$ & $0$ if $p=3$, otherwise $2(p-1)/p^2$ \\
Gateway $b$ in $K_q$ & $0$ if $q=3$, otherwise $2(q-1)/q^2$ \\
Non-gateway vertex in $K_p$ or $K_q$ & $0$ \\
Leaf vertex of $s$ & outside $V_{\ge2}$ (degree 1) \\
\bottomrule
\end{tabular}
\end{table}

\textit{Proof.} Table~\ref{tab:hub_scores} leaves only the two clique gateways as potential positive competitors. For $m\ge4$,
\[
\frac{2(m-1)}{m^2}\le\frac{3}{8}<\frac{2}{3},
\]
and the gateway score is zero when $m=3$. Hence $x$ is uniquely maximal. The parameter $r$ appears neither in the score of $x$ nor in that of any competing vertex with positive score.

A natural simpler alternative is to combine low focal clustering with degree. The same graph shows why that combination is insufficient for the structural role targeted by the IC.

\begin{proposition}[Separation from degree-weighted low clustering]
In $H_{p,q,r}$, the factor $1-C(v)$ assigns the same value to the cohesive connector $x$ and the cohesion-free hub $s$. If low clustering is weighted directly by degree, then
\[
d(x)\bigl(1-C(x)\bigr)=3,
\qquad
d(s)\bigl(1-C(s)\bigr)=r+1.
\]
Thus, for every $r\ge3$, the degree-weighted low-clustering score ranks $s$ above $x$, whereas $\IC(x)=2/3>0=\IC(s)$ for all $r\ge1$.
\end{proposition}
\textit{Proof.} Both $x$ and $s$ have zero focal clustering, so $1-C(x)=1-C(s)=1$. Their degrees are $d(x)=3$ and $d(s)=r+1$, which gives the displayed values. The IC separates the vertices because its second factor evaluates cohesion behind adjacent vertices. Two of the three directions incident to $x$ open into complete branches. Every direction beyond $s$ is cohesion-free.

This construction does not claim that degree is generally irrelevant. It isolates one specific perturbation in which arbitrarily many leaves are added behind an intermediate hub without creating any new cohesive branch. The result gives an exact counterexample to the idea that low focal clustering, even when weighted by degree, is sufficient to encode the connector role targeted by the IC.

\subsection{Density-Controlled Windmill Branches}
The preceding family keeps both cohesive branches complete. To vary branch cohesion in a controlled way, consider a windmill construction. Let $c$ be the number of mutually disjoint complete leaf groups behind a gateway, and let $h\ge3$ be the number of vertices in each group. Thus, $ch$ leaf-group vertices lie behind the gateway. A \emph{windmill branch} $W_{c,h}$ contains a gateway vertex $a$ and $c$ pairwise disjoint leaf groups $P_1,\ldots,P_c$ of order $h$, such that each $\{a\}\cup P_j$ induces $K_{h+1}$ and no edges connect distinct leaf groups. The graph $J_{c,h,r}$ consists of two such windmill branches with gateway vertices $a$ and $b$, a vertex $x$ adjacent to $a$, $b$, and $s$, and $r\ge1$ leaves adjacent to $s$.

Figure~\ref{fig:windmill} shows both windmill branches explicitly.

\begin{figure}[H]
\centering
\begin{tikzpicture}[scale=.68,every node/.style={font=\scriptsize}]
\node[circle,draw,inner sep=1.4pt] (a) at (-2.4,0) {$a$};
\node[circle,draw,very thick,inner sep=1.4pt] (x) at (0,0) {$x$};
\node[circle,draw,inner sep=1.4pt] (b) at (2.4,0) {$b$};
\draw (a)--(x)--(b);
\node[draw,rounded corners,minimum width=1.85cm,minimum height=.68cm] (L1) at (-5.0,1.35) {$P_1\cong K_h$};
\node[draw,rounded corners,minimum width=1.85cm,minimum height=.68cm] (L2) at (-5.0,0) {$P_2\cong K_h$};
\node[draw,rounded corners,minimum width=1.85cm,minimum height=.68cm] (Lc) at (-5.0,-1.35) {$P_c\cong K_h$};
\draw (a)--(L1.east) (a)--(L2.east) (a)--(Lc.east);
\node at (-5.0,-.72) {$\vdots$};
\node[draw,rounded corners,minimum width=1.85cm,minimum height=.68cm] (R1) at (5.0,1.35) {$Q_1\cong K_h$};
\node[draw,rounded corners,minimum width=1.85cm,minimum height=.68cm] (R2) at (5.0,0) {$Q_2\cong K_h$};
\node[draw,rounded corners,minimum width=1.85cm,minimum height=.68cm] (Rc) at (5.0,-1.35) {$Q_c\cong K_h$};
\draw (b)--(R1.west) (b)--(R2.west) (b)--(Rc.west);
\node at (5.0,-.72) {$\vdots$};
\node[draw,rectangle,inner sep=2pt] (s) at (0,-1.70) {$s$};
\draw (x)--(s);
\foreach \xx in {-0.75,0,.75} {\node[circle,draw,inner sep=.8pt] at (\xx,-2.65) {}; \draw (s)--(\xx,-2.65);}
\node at (1.35,-2.65) {$\cdots\ r$ leaves};
\end{tikzpicture}
\caption{Density-controlled windmill family $J_{c,h,r}$. Both gateway vertices $a$ and $b$ connect to $c$ mutually disconnected complete leaf groups of order $h$. Increasing $c$ fragments each branch-side neighborhood while preserving complete cohesion within every individual group.}
\label{fig:windmill}
\end{figure}
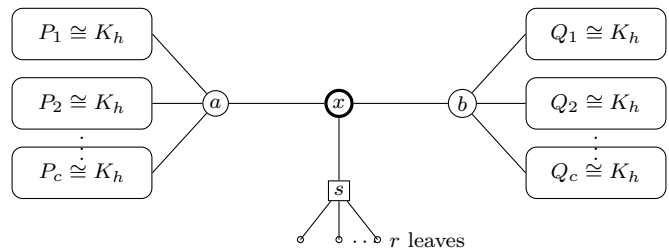

From the perspective of $x$, the leave-one-out cohesion behind either gateway is
\begin{equation}
C_{a\setminus x}=C_{b\setminus x}
=\frac{c\binom{h}{2}}{\binom{ch}{2}}
=\frac{h-1}{ch-1}.
\end{equation}
The numerator counts the edges present inside the $c$ complete groups. The denominator counts all possible edges among the $ch$ vertices in the branch-side neighborhood. For fixed $h$, this value equals one when $c=1$ and decreases monotonically toward zero as the branch is split into more mutually disconnected complete groups.

Vertex $x$ has the three leave-one-out terms $(h-1)/(ch-1)$, $(h-1)/(ch-1)$, and $0$, with $C(x)=0$. A gateway such as $a$ has $d(a)=ch+1$ and
\[
C(a)=\frac{2c\binom{h}{2}}{(ch+1)ch}=\frac{h-1}{ch+1},
\qquad
1-C(a)=\frac{(c-1)h+2}{ch+1}.
\]
Its $ch$ leaf-group neighbors contribute leave-one-out value one, while $x$ contributes zero. Hence
\[
\overline C_{\mathrm{LOO}}(a)=\frac{ch}{ch+1}.
\]
The relevant IC values are therefore
\begin{align}
\IC(x)&=\frac{2(h-1)}{3(ch-1)},\\
\IC(a)=\IC(b)&=\frac{ch\bigl((c-1)h+2\bigr)}{(ch+1)^2},\\
\IC(s)&=0.
\end{align}
Every vertex within a leaf group has IC value zero because its ordinary clustering coefficient equals one.

\begin{proposition}[Fragmentation threshold]
For all integers $c\ge1$, $h\ge3$, and $r\ge1$, the unique IC maximizer is $x$ when $c=1$, whereas the maximizer set is $\{a,b\}$ when $c\ge2$.
\end{proposition}
\textit{Proof.} Only $x,a,b$ have positive IC values. For $c=1$, their scores reduce to $2/3$ and $2h/(h+1)^2$, respectively. The inequality $2h/(h+1)^2<2/3$ is equivalent to $3h<(h+1)^2$, which holds because $(h+1)^2-3h=h^2-h+1>0$. For $c\ge2$,
\[
\IC(x)=\frac{2(h-1)}{3(ch-1)}<\frac13,
\]
because $2(h-1)<ch-1$. At the same time,
\[
2ch\bigl((c-1)h+2\bigr)-(ch+1)^2
=c(c-2)h^2+2ch-1>0,
\]
so $\IC(a)=\IC(b)>1/2$. Thus, for all $c\ge2$, both $a$ and $b$ strictly exceed $x$ and are the only maximizers.

The transition has a direct structural interpretation. For $c=1$, deleting $x$ separates the two complete branch remainders, so $x$ is an articulation vertex in this calibration graph. Once $c\ge2$, each gateway itself opens into several complete leaf groups. The stronger semi-local decomposition then occurs at $a$ and $b$, and the IC maximum shifts inward. The threshold is an exact change in the local organization represented by the graph family.

\section{Application to the Human Interactome Map}
The protein--protein interaction network provides an empirical application of the IC to a large biological graph. The IC was previously applied to a Human Interactome Map \cite{Wiebringhaus2009} and later presented in the context of intermodular connections \cite{WiebringhausBrinck2014}.

\subsection{Network and Summary Statistics}
The application uses a Human Interactome Map (HIM) \cite{Wiebringhaus2009}. The graph used for the IC analysis contains 9,222 vertices and 36,324 edges. Its mean ordinary local clustering coefficient is
\begin{equation}
\langle C\rangle = 0.114.
\end{equation}

\subsection{Degree-Dependent Stabilization of the IC}
Figure~\ref{fig:him_degree_ic} shows the degree--IC distribution. Three features are visually apparent. First, low-degree vertices span several orders of magnitude in IC. Second, the distribution narrows with increasing degree and approaches a band slightly above $0.1$. Third, the isolated rightmost hub, at a degree of roughly 750, lies close to the background level, whereas the rectangle marks several positive deviations at degrees around 20--30.

\begin{figure}[H]
\centering
\includegraphics[alt={Degree--IC scatter plot for the Human Interactome Map, with high-degree vertices narrowing toward a low IC level.},width=.78\textwidth]{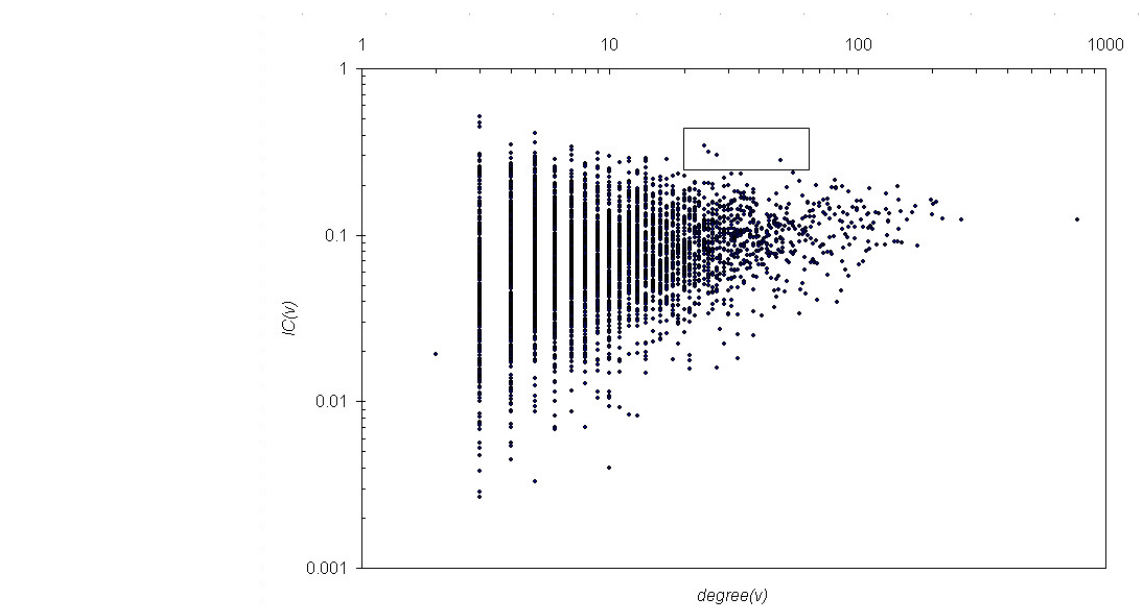}
\caption{Degree--IC distribution of the Interconnectedness Coefficient (IC) in the Human Interactome Map \cite{Wiebringhaus2009}. Both axes are logarithmic. The point cloud narrows toward the background level near $0.114$ as degree increases. The rectangle marks the selected positive deviations considered in the subsequent biological interpretation. Vertices with $\IC=0$ are not visible on the logarithmic ordinate, and vertices with identical or nearly identical coordinates may overlap.}
\label{fig:him_degree_ic}
\end{figure}

The product representation in Eq.~\eqref{eq:ic_product} makes the high-degree behavior directly visible.

\Needspace{10\baselineskip}
\begin{proposition}[High-degree stabilization]
Let $(v_n)$ be a sequence of vertices with $d(v_n)\to\infty$. If
\[
C(v_n)\to c_\infty
\qquad\text{and}\qquad
\overline C_{\mathrm{LOO}}(v_n)\to\mu_\infty,
\]
then
\begin{equation}
\IC(v_n)\longrightarrow (1-c_\infty)\mu_\infty.
\end{equation}
In particular, if focal clustering of high-degree vertices tends to zero, then $\IC(v_n)\to\mu_\infty$.
\end{proposition}
\textit{Proof.} The result follows immediately from the product representation of the IC and continuity of multiplication.

The condition $d(v_n)\to\infty$ does not determine the limit of $C(v_n)$. In particular, increasing degree alone implies neither $C(v_n)\to1$ nor $C(v_n)\to0$. The relevant regime for the stabilization considered here is one in which focal clustering becomes small for high-degree vertices, so that $1-C(v_n)\to1$. In the PPI network analyzed here, this regime includes high-degree vertices, i.e. hubs, with low local clustering. The role of the factor $1-C(v)$ is therefore directly visible.

This elementary limit provides a useful interpretation of Figure~\ref{fig:him_degree_ic}. In a hierarchical network regime in which high-degree vertices are only weakly clustered, the focal sparsity factor $1-C(v)$ approaches one. The IC is then increasingly determined by the mean branch-side cohesion of adjacent vertices. If this mean stabilizes, the score stabilizes as well even while degree continues to increase. Normalization by $d(v)$ converts increasing neighborhood size into an average rather than an extensive reward for high degree.

The numerical level is informative as well. The network-wide mean clustering coefficient of the HIM is $0.114$. If the branch-side mean at high degree is approximated by this characteristic clustering level of the network,
\[
\mu_\infty \approx \langle C\rangle = 0.114,
\]
then the limiting relation predicts
\begin{equation}
\boxed{\IC_{\text{high degree}}\approx 0.114,}
\end{equation}
which is consistent with the plateau visible in the empirical plot. This equality is an approximation rather than an identity implied by the definition, because $\langle C\rangle$ averages ordinary clustering over all vertices, whereas $\overline C_{\mathrm{LOO}}$ averages leave-one-out clustering over neighbors. As detailed in Appendix~\ref{app:limit}, sampling vertices through incident edges is degree-biased and therefore does not in general reproduce the ordinary vertex-wise clustering mean.

Equivalently, under the joint assumptions $c_\infty=0$ and $\mu_\infty\approx\langle C\rangle>0$, the same calibration can be expressed as
\begin{equation}
\frac{\IC(v_n)}{\langle C\rangle}\longrightarrow 1,
\qquad
\IC(v_n)-\langle C\rangle\longrightarrow 0.
\end{equation}
For finite-degree vertices, the difference is positive or negative according to whether a score lies above or below the chosen background level. Neither expression supplies an additional theorem about the empirical network. Both are rescalings of the stated approximation.

The simultaneous narrowing of the point cloud has the same qualitative origin. The IC of a high-degree vertex contains an average over many branch-side terms. Under weak dependence and finite variance, the variability of such a mean decreases on the usual $d(v)^{-1/2}$ scale. The observed distribution is therefore consistent both with convergence of the mean level and with concentration around that level.

\subsection{Biological Interpretation of High-IC Proteins}
The biological interpretation of proteins with high IC values reflects established interface, scaffold, and adaptor functions. Table~\ref{tab:high_ic} decomposes four high-degree examples into the two factors of the IC, making the local signature directly visible.

\begin{table}[H]
\centering
\caption{Selected high-degree vertices within the top IC decile and decomposition of the IC into its two factors \cite{Wiebringhaus2009}.}
\label{tab:high_ic}
\begin{tabular}{lrrrrr}
\toprule
Protein & $d(v)$ & $C(v)$ & $1-C(v)$ & $\overline C_{\mathrm{LOO}}(v)$ & IC \\
\midrule
MED29 (IXL) & 25 & 0.247 & 0.753 & 0.421 & 0.317 \\
MED28 (EG1) & 24 & 0.250 & 0.750 & 0.456 & 0.342 \\
MED9 & 27 & 0.185 & 0.815 & 0.373 & 0.304 \\
ORC2L & 22 & 0.355 & 0.645 & 0.292 & 0.188 \\
\bottomrule
\end{tabular}
\end{table}

Three of the four proteins belong to the Mediator complex. MED29 (IXL) and MED28 (EG1) are located in the Tail region at the structural coupling to the conserved Mediator core. The high-resolution mammalian Mediator structure reveals an extended MED27--MED30 connection between the core and Lower Tail, together with direct Tail--Head contacts involving MED28 \cite{ZhaoMediator2021}. MED9 is a component of the Middle module, where the Med7C--Med21--Med4--Med9 tetramer forms a structural scaffold for recruitment of the remaining Middle-module subunits \cite{Robinson2015}. The three Mediator hits therefore constitute a coherent example of interface and scaffold positions within a complex whose biological function itself consists in mediating between regulatory activators and RNA polymerase II.

ORC2L, now designated ORC2, forms part of the structural core of the Origin Recognition Complex. ORC binds replication origins and, together with CDC6, provides the platform for recruitment of the MCM2--7 complex during pre-replication-complex assembly. Structural data show how defined interaction surfaces convert ORC--CDC6 into the active recruiter of MCM2--7 \cite{Feng2021}. ORC2L therefore likewise represents a scaffold and platform function within a modular initiation machinery.

Additional highly ranked IC proteins exhibit the same functional direction. SKIV2L2/\allowbreak MTR4 (MTREX) couples nuclear RNA adaptors to the RNA exosome. Its Arch domain serves as a binding platform for distinct adaptor proteins and links RNA recognition to exosomal processing \cite{Lingaraju2019}. RBBP1/\allowbreak ARID4A, recorded as RBP1 in the interaction annotation, recruits the histone deacetylase complex to retinoblastoma-\allowbreak family proteins within the mSIN3/\allowbreak HDAC system and has been experimentally characterized as an adaptor protein \cite{Lai2001}.

The PPI application therefore provides an empirical demonstration of the structural interpretation of the IC. Among the highly ranked vertices are proteins with established connector, interface, scaffold, and adaptor roles. The recovery of this pattern in an incompletely observed protein-interaction network is consistent with the practical interpretability of the semi-local score beyond the analytical graph families. It is not an independent validation based on a newly reconstructed network.

\section{Discussion}
\textbf{What the IC measures.} The IC is a bounded structural contrast between low cohesion among the focal neighbors and high cohesion immediately beyond those neighbors. Its construction admits a direct equality case. A score of one is attained only when the focal neighborhood is completely unclustered and every adjacent branch remainder is complete. This gives the numerical scale a sharper interpretation than a purely ordinal ranking score. The IC is dimensionless and bounded between zero and one. A value of $0.317$ therefore does not represent $31.7$ percent. It is a node weight on the defined structural scale. Semantically, the IC can be interpreted as a semi-local measure of connector vertices between cohesive regions. The high-degree result provides a second calibration. When focal clustering vanishes and mean branch-side cohesion stabilizes, the numerical IC level converges to that structural background level.

\textbf{What the IC does not measure.} High degree, epidemic spreading ability, global control of shortest paths, community participation, and articulation are not built into the score. Some of these roles may correlate with the IC in particular networks, but they do not define its objective. Different network questions can therefore induce legitimately different rankings.

\textbf{Semi-local screening versus global topology.} The absence of an articulation requirement is both an advantage and a limitation. It preserves the bounded information radius and reduces dependence on long alternative paths. A high IC value cannot establish that deleting the vertex disconnects the graph. An early complementary node-cut analysis found that iterative removal of an upper fraction of IC-ranked vertices was accompanied by a decrease in the mean local clustering coefficient \cite{Wiebringhaus2009}. This preliminary observation motivates systematic perturbation analysis of vertices prioritized by the IC.

\textbf{Relation to clustering-based influence measures.} Berahmand et al. use closely related constituent information but address a different objective \cite{Berahmand2018}. Their score was introduced to identify influential spreaders and was evaluated using diffusion models. The IC instead screens for connectors between cohesive regions. A quantitative comparison therefore requires an evaluation criterion that is independent of either construction. Diffusion performance would directly favor spreading-oriented measures. A criterion defined from the IC's own branch-cohesion structure would favor the IC. Future comparisons should therefore use external structural outcomes.

\textbf{Computational scope.} The IC uses radius-two topology and can be computed from triangle counts and common-neighbor intersections. It does not require all-pairs or source-to-all shortest-path information. This makes the measure computationally attractive for large graphs in settings where exact global path-based centralities become expensive. The point is not that the IC replaces betweenness centrality. The two measures represent different structural questions. The computational advantage follows from the bounded information radius of the IC.

\textbf{Applications beyond biology.} The same semi-local pattern may represent a broker between cohesive social or organizational groups, a researcher linking collaboration communities, a software component between internally cohesive packages, or a concept bridging established domains in knowledge and patent networks. These interpretations require undirected one-mode graphs in which triangles carry substantive information about local cohesion. The corollary for triangle-free graphs makes the boundary exact. On trees and simple bipartite graphs, $\IC\equiv0$. In networks with structurally very few triangles, the unmodified score is correspondingly weakly discriminative. Directed and weighted variants require separate definitions.

\section{Future Work}
A first extension is systematic targeted node-attack analysis across several real-world networks. Such experiments can compare removal by IC ranking with random removal and established local or global ranking measures, while evaluating changes in fragmentation, connectivity, and local cohesion. This would directly test whether vertices prioritized by the semi-local IC correspond to structurally vulnerable positions under perturbation.

A second direction is to characterize graph classes on which rankings induced by the IC and simpler comparison scores coincide or necessarily diverge. The hub-extended bridge already gives an exact separation from $1-C(v)$ and from direct degree weighting of low clustering. More general results could clarify the specific information contributed by branch-side leave-one-out cohesion. Ranking stability under edge errors and incomplete observation is a related problem.

Further theoretical extensions include a group-level IC for cohesive subgraphs, cliques, or communities, weighted and directed variants, and asymptotic results for random and hierarchical graph models. In particular, the expectation, variance, and concentration of the IC in controlled network ensembles are natural theoretical continuations of the high-degree stabilization observed here.

\section{Conclusion}
The Interconnectedness Coefficient (IC) is a bounded, semi-local, partition-free graph-theoretic node measure for connector vertices between cohesive network regions. Its definition separates low focal clustering from high leave-one-out cohesion beyond adjacent vertices. The range theorem establishes the scale and its exact maximum. Complete multi-branch constructions show that the maximum is invariant to the number and size of perfectly cohesive branches, while a one-edge perturbation gives an explicit gradual decrease from the ideal value. The hub-extended bridge family $H_{p,q,r}$ shows that a large cohesion-free hub does not displace the connector between cohesive regions. The same family also shows that direct degree weighting of low clustering can rank the cohesion-free hub above that connector. The windmill family $J_{c,h,r}$ identifies an exact fragmentation threshold at which the stronger semi-local connector role shifts from the external vertex to the gateway vertices of the branches.

The computation reduces to triangle counts and common-neighbor intersections and requires no global shortest-path information. The Human Interactome Map application adds a further property. At high degree, IC values narrow toward a level close to the network's mean clustering coefficient of $0.114$. The product representation explains this behavior. As focal clustering becomes negligible, the IC approaches the mean leave-one-out cohesion of adjacent branches. Under the network-level approximation that this branch-side mean is close to the characteristic clustering level, the observed plateau near $0.114$ follows directly.

The PPI application complements the theory with a real-world demonstration. Highly ranked vertices include proteins with established interface, scaffold, and adaptor roles, notably the Mediator hits MED29, MED28, and MED9, together with ORC2, MTR4, and RBBP1/\allowbreak ARID4A. The IC is therefore structurally interpretable beyond synthetic graph families as a semi-local connector measure for cohesive network regions.

\appendix
\section{Technical Note on the High-Degree Limit}\label{app:limit}
The empirical approximation $\IC_{\text{high degree}}\approx0.114$ must be distinguished from a universal identity. The global ordinary mean clustering coefficient is
\begin{equation}
\langle C\rangle=\frac{1}{|V|}\sum_{v\in V}C(v),
\end{equation}
whereas the IC contains a neighborhood-conditioned leave-one-out mean. Even in an uncorrelated network, sampling a vertex through an incident edge approximately weights degree classes according to
\begin{equation}
Q(k)=\frac{kP(k)}{\langle k\rangle},
\end{equation}
rather than according to the vertex distribution $P(k)$. In addition, $C_{u\setminus v}$ differs slightly from ordinary $C(u)$ because the focal connection is excluded. Consequently, $\mu_\infty=\langle C\rangle$ requires an approximation concerning the composition of high-degree neighborhoods. This equality does not follow from the definition of the IC alone.

The concentration visible in Figure~\ref{fig:him_degree_ic} also admits a simple sampling-theoretic interpretation. If the branch-side terms $X_i=C_{u_i\setminus v}$ have stable variance $\sigma^2$ and are independent or only weakly dependent, then approximately
\begin{equation}
\operatorname{Var}\!\left(\frac{1}{d(v)}\sum_{i=1}^{d(v)}X_i\right)\approx\frac{\sigma^2}{d(v)},
\end{equation}
so that the standard deviation decreases approximately as $d(v)^{-1/2}$. It describes the mechanism by which averaging over an increasing degree can produce the observed narrowing.

\vspace{4pt}
\noindent{\scriptsize Generative language models assisted with manuscript preparation. All scientific responsibility remains with the author.}

\end{document}